\documentclass[11pt]{article}
\usepackage[letterpaper,margin=1in]{geometry}
\usepackage{amsmath, amssymb, amsthm, amsfonts}
\usepackage{cite}

\usepackage{bbm}

\usepackage{ifthen}
\usepackage{tikz}
\usetikzlibrary{positioning,decorations.pathreplacing}

\usepackage{appendix}
\usepackage{graphicx}
\usepackage{color}
\usepackage{algorithm}
\usepackage[noend]{algpseudocode}
\usepackage{epstopdf}
\usepackage{wrapfig}
\usepackage{thm-restate}

\usepackage{wasysym}

\newcommand{\polylog}{\mathrm{polylog}}

\usepackage{framed}
\usepackage[framemethod=tikz]{mdframed}
\usepackage[bottom]{footmisc}
\usepackage{enumitem}
\setitemize{noitemsep,topsep=3pt,parsep=3pt,partopsep=3pt}
\usepackage[font=small]{caption}
\usepackage{xspace}

\newtheorem{theorem}{Theorem}[section]
\newtheorem{lemma}[theorem]{Lemma}
\newtheorem{meta-theorem}[theorem]{Meta-Theorem}

\newtheorem{corollary}[theorem]{Corollary}

\newtheorem{observation}[theorem]{Observation}
\newtheorem{definition}[theorem]{Definition}

\definecolor{darkgreen}{rgb}{0,0.5,0}
\usepackage{hyperref}
\hypersetup{
    unicode=false,          
    colorlinks=true,        
    linkcolor=red,          
    citecolor=darkgreen,        
    filecolor=magenta,      
    urlcolor=cyan           
}
\usepackage[capitalize, nameinlink]{cleveref}
\Crefname{lemma}{Lemma}{Lemmas}
\Crefname{claim}{Claim}{Claims}
\Crefname{remark}{Remark}{Remarks}
\Crefname{observation}{Observation}{Observations}
\algnewcommand\algorithmicswitch{\textbf{switch}}
\algnewcommand\algorithmiccase{\textbf{case}}

\algdef{SE}[SWITCH]{Switch}{EndSwitch}[1]{\algorithmicswitch\ #1\ \algorithmicdo}{\algorithmicend\ \algorithmicswitch}%
\algdef{SE}[CASE]{Case}{EndCase}[1]{\algorithmiccase\ #1}{\algorithmicend\ \algorithmiccase}%
\algtext*{EndSwitch}%
\algtext*{EndCase}%

\newcommand{\eps}{\varepsilon}

\newcommand{\energy}{$\mathsf{ENERGY}$\xspace}
\newcommand{\congest}{$\mathsf{CONGEST}$\xspace}
\newcommand{\local}{$\mathsf{LOCAL}$\xspace}

\newcommand{\poly}{\operatorname{\text{{\rm poly}}}}

\let\oldtextbf=\textbf
\renewcommand\textbf[1]{{\boldmath\oldtextbf{#1}}}

\newcommand{\FullOrShort}{short}

\ifthenelse{\equal{\FullOrShort}{full}}{
	
  \newcommand{\fullOnly}[1]{#1}
  \newcommand{\shortOnly}[1]{}
}{
    \newcommand{\shortOnly}[1]{#1}
    \newcommand{\fullOnly}[1]{}
  }

\begin{document}
\date{}
\title{A Near-Optimal Low-Energy Deterministic Distributed SSSP \\ with Ramifications on Congestion and APSP}

\author{
  Mohsen Ghaffari\\
  \small MIT \\
  \small ghaffari@mit.edu
  \and 
  Anton Trygub\\
  \small MIT \\
  \small trygub@mit.edu
 }

\maketitle
\begin{abstract} 
We present a low-energy deterministic distributed algorithm that computes exact Single-Source Shortest Paths (SSSP) in near-optimal time: it runs in $\tilde{O}(n)$ rounds and each node is awake during only $\poly(\log n)$ rounds. When a node is not awake, it performs no computations or communications and spends no energy.

The general approach we take along the way to this result can be viewed as a novel adaptation of Dijkstra's classic approach to SSSP, which makes it suitable for the distributed setting. Notice that Dijkstra's algorithm itself is not efficient in the distributed setting due to its need for repeatedly computing the minimum-distance unvisited node in the entire network. Our adapted approach has other implications, as we outline next.

As a step toward the above end-result, we obtain a simple deterministic algorithm for exact SSSP with near-optimal time and message complexities of $\tilde{O}(n)$ and $\tilde{O}(m)$, in which each edge communicates only $\poly(\log n)$ messages. Therefore, one can simultaneously run $n$ instances of it for $n$ sources, using a simple random delay scheduling. That computes All Pairs Shortest Paths (APSP) in the near-optimal time complexity of $\tilde{O}(n)$. This algorithm matches the complexity of the recent APSP algorithm of Bernstein and Nanongkai [STOC 2019] using a completely different method (and one that is more modular, in the sense that the SSSPs are solved independently). It also takes a step toward resolving the open problem on a deterministic $\tilde{O}(n)$-time APSP, as the only randomness used now is in the scheduling.
\end{abstract}






\section{Introduction and related work}
This paper is centered on distributed algorithms for the Single-Source Shortest Paths (SSSP) problem, i.e., computing the exact distances and shortest paths in a computer network from a designated source node. This is one of the basic and widely used problems in graph algorithms and especially computer networks. Our initial focus is on a relatively new facet: energy-efficient distributed algorithms. We believe this is practically relevant and theoretically interesting. But, as we outline soon, the approach we develop ends up having ramifications for more classic facets of the problem, namely the congestion of the SSSP algorithm and the extension to All-Pairs Shortest Paths (APSP).

Concretely, our primary objective is distributed SSSP algorithms with low energy. Informally, we want algorithms where each node is awake and active (computing or communicating) during only a ``small'' amount of time and therefore it consumes only a small amount of \textit{energy}. We present the actual definitions of the energy complexity and more discussions on it later. As is standard, the interpretation of small here is bounds that are at most $\poly(\log n)$, where $n$ denotes the number of nodes. As our main end result, we present a $\poly(\log n)$-energy deterministic distributed algorithm for exact SSSP with a near-optimal time complexity of $\tilde{O}(n)$. The energy bound means each node is awake for at most $\poly(\log n)$ time. This makes the algorithm scalable in terms of each node's energy consumption and thus suitable for energy-constrained networks, e.g., sensor networks. 
 Moreover, the energy bound directly implies a \textit{congestion} bound: since each node communicates in only $\poly(\log n)$ rounds, the algorithm sends at most $\poly(\log n)$ messages through each edge. This makes the algorithm suitable for networking settings that run many algorithms concurrently, as is usually the case in most computer networks. This is because low congestion algorithms allow the network to run many of them simultaneously, without significant slow down~\cite{ghaffari2015scheduling}. 
 
 This $\tilde{O}(n)$ time complexity is trivially near-optimal in worst-case graphs, with diameter $\tilde{\Theta}(n)$. It is also near-optimal in a stronger sense: an $\tilde{\Omega}(n)$ lower bound holds in graphs with diameter $D=\poly(\log n)$, once we insist on congestion $\poly(\log n)$. This is known and can be seen by simple adaptations of the well-known lower bound of Das Sarma et al.~\cite{DasSarma2011DistributedApproximation}. See also \cite[Section 10.6]{ghaffari2017thesis} for more on this, including pointing out the same (though for the lower bound applied to the MST problem), and for general discussions about low congestion algorithms and their benefits.

Our algorithm might be of interest not only for its end results but also more broadly for its novel approaches. The algorithm has two key ingredients: 
\begin{itemize}
\item[(I)] a new approach to the distributed computation of exact SSSP, which can be viewed as a distributed analog of Dijkstra's SSSP algorithm, and 
\item[(II)] an energy-efficient \textit{deterministic} distributed algorithm for unweighted SSSP (i.e., BFS).
\end{itemize}

Ingredient (I) has ramifications beyond energy considerations: 
It gives a $\poly(\log n)$-congestion $\tilde{O}(n)$ deterministic algorithm for SSSP. That directly yields an $\tilde{O}(n)$-time APSP algorithm via random scheduling~\cite{Leighton1994PacketSteps,ghaffari2015scheduling}, hence matching the near-optimal time-complexity that was achieved by Bernstein and Nanongkai~\cite{bernstein2019distributed}, after a line of work~\cite{elkin2017distributed,huang2017distributed}. Our approach is completely different and has some additional benefits. 

Ingredient (II) is based on a local awake/sleep coordination mechanism via deterministically constructed neighborhood covers, with a low-energy construction. It provides a deterministic counterpart to the randomized low-energy BFS algorithm of Dani and Hayes~\cite{dani2022wake}. We think that these deterministic BFS and neighborhood cover constructions might find applications in deterministic energy-efficient distributed algorithms for other graph problems. 

Next, we review the models and state our concrete contributions. To make the first ingredient more widely accessible, in \Cref{subsec:WithoutEnergy}, we initially ignore energy considerations and discuss the new distributed approach to SSSP and its ramifications. Then, in \Cref{subsec:WithEnergy}, we zoom in on energy complexity and discuss our energy-efficient and deterministic BFS and SSSP algorithms. 

\subsection{Without Energy Considerations} 
\label{subsec:WithoutEnergy}

We first recall the message-passing model of distributed computation. Then we outline our new approach to SSSP in the context of the classic approaches and the recent related work.

\paragraph{Synchronous Distributed Message-Passing Model.} The network is abstracted as an undirected weighted graph $G=(V, E)$, and we use the notations $n:=|V|$ and $m:=|E|$. It is usually assumed that the weight $w(e)$ of each edge $e\in E$ is in the range $[1, \poly(n)]$. Each node represents one computer/processor, equipped with a unique identifier, typically assumed to have $b=O(\log n)$ bits. In this synchronous model, computation and communications occur in lock-step rounds $1, 2, 3, \dots$. Per round, each node performs some computations on the data that it holds. Then, it can send one $B$-bit message to each of its neighbors. All the messages sent in a round arrive by the end of the round. Then, the algorithm proceeds to the next round. Typically, we assume $B=O(\log n)$. This model with this message size is sometimes called \congest~\cite{Peleg2000DistributedApproach}. The \textit{time complexity} is the number of rounds until all nodes compute their outputs. The \textit{message complexity} is the total number of messages sent during the algorithm throughout the network.


\paragraph{SSSP in the Distributed Setting:} 
To get an understanding of the problem, let us revisit some of the classic approaches and their shortcomings. For SSSP, there is a simple algorithm of Bellman and Ford~\cite{bellman1958routing, ford1956network} that runs in the optimal $O(n)$ time: Start with $\bar{d}(s, s)=0$ and $\bar{d}(s, u)=\infty, \forall u\neq s$. Then, per round, each node $u$ updates its distance estimate $\bar{d}(s, u)=\min_{v \in N(u)} \big(\bar{d}(s, v) + w(vu)\big)$. Here, $w(vu)$ denotes the weight of the edge connecting neighbor $v$ to node $u$. In essence,  $\bar{d}(s, u)$ is an estimate on the distance, and it is always maintained to satisfy $\bar{d}(s, u)\geq dist(s, u)$. Node $u$ updates its estimate $\bar{d}(s, u)$ by checking each neighbor $v$ and seeing if connecting through $v$, with its current distance estimate $\bar{d}(s, v)$, gives $u$ a shorter path from the source $s$. This operation is called \textit{relaxing edge $vu$}. A major drawback is that this algorithm relaxes each edge in each round, and thus has message complexity $\Theta(mn)$, and $\Omega(n)$ congestion--this becomes problematic when trying to run several SSSPs concurrently. See, e.g., \cite[Sections 2.1 \& 3.1]{nanongkai2014distributed}.

The above $\Theta(mn)$ number of relaxations is also too expensive for sequential SSSP computations. In the sequential setting, for undirected graphs, the classic algorithm of Dijkstra remedies this and achieves an $\tilde{O}(m)$ complexity. The key is to choose the relaxations carefully such that each edge is relaxed only once. Dijkstra's algorithm maintains a monotonically growing set $T$ of ``\textit{visited}" nodes, where each node $u\in T$ already knows its distance $\bar{d}(s, u)=dist(s, u)$. Each unvisited node $v\in V\setminus T$ remains with a distance estimate $\bar{d}(s, v)$ that is not necessarily finalized. Dijkstra then finds the unvisited node in $V\setminus T$ that has the smallest $\bar{d}(s, v)$, and adds $v$ to $T$, thus visiting it. It then relaxes all edges connecting $v$ to unvisited nodes. This way, each edge is relaxed only once. This approach is not suitable for distributed computing, because finding in each iteration the global minimizer of $\bar{d}(s, v)$ requires extra time and messages. A direct distributed implementation of Dijkstra would have time complexity $O(nD)$, where $D$ is the hop diameter and can be as large as $\Theta(n)$, and message complexity $O(n^2+m)$. These are far away from the desired bounds. 

Over the past decade, there has been much progress on distributed SSSP and its approximate variant~\cite{lenzen2013fast,nanongkai2014distributed, henzinger2016deterministic, elkin2017distributed, ghaffari2018improved,forster2018faster, chechik2020single,rozhovn2022undirected,rozhovn2023parallel}. The focus was on sublinear time complexities in graphs with sublinear hop diameter. In our case, we want congestion (and energy) bounds of $\poly(\log n)$, and thus we cannot hope to achieve a time complexity faster than $\tilde{\Omega}(n)$ even in low-diameter graphs (see \cite[Section 10.6]{ghaffari2017thesis}). Among these prior works, of particular relevance to our paper is an algorithm of Nanongkai~\cite{nanongkai2014distributed}, which computes a $(1+\eps)$ approximations of SSSP in the synchronous model of distributed computing, for any constant $\eps>0$, with $\tilde{O}(n)$ time complexity and $\tilde{O}(m)$ message complexity. In a sense, this algorithm is a simple and elegant \textit{rounding scheme} that reduces the $(1+\eps)$ approximations of SSSP to $O(\log n)$ instances of BFS (i.e., unweighted SSSP) in an undirected graph with ${O}(n)$ nodes. \footnote{We comment that several of the other works on exact algorithms use approximate versions as subroutines, via the scaling framework, see e.g., \cite{ghaffari2018improved,forster2018faster,chechik2020single,rozhovn2023parallel}. However, the transformations in this framework turn the graph into directed graphs (more accurately, graphs with asymmetric weights along the two directions of the edge), and this seems to cause a major obstacle in using any of those ideas with a low energy.}

\paragraph{Our result for low-congestion SSSP.} Our approach gives a deterministic distributed algorithm that computes SSSP with near-optimal time and message complexities: $\tilde{O}(n)$ time and $\tilde{O}(m)$ messages\footnote{This time complexity optimality is in the more classic worst-case graph sense considered in the distributed algorithms literature, where the diameter can be large~\cite{awerbuch1987optimal}. It would be interesting to improve the complexity in low-diameter graphs.} More crucially, the algorithm has only $\poly(\log n)$ congestion, that is, each edge sends at most $\poly(\log n)$ messages throughout the execution of the algorithm.

\paragraph{Implications for APSP} Considering the $\poly(\log n)$ congestion of our algorithm per edge, we can run $n$ copies of it for different sources concurrently in $\tilde{O}(n)$ time, using the simple random delays idea of scheduling~\cite{Leighton1994PacketSteps}, e.g., with the black-box theorems in~\cite{ghaffari2015scheduling}. Prior to our work, there was a progression of improvements on APSP, with complexities $\tilde{O}(n^{5/3})$ by Elkin~\cite{elkin2017distributed}, $\tilde{O}(n^{5/4})$ by Huang et al.\cite{huang2017distributed}, and then essentially resolved to $\tilde{O}(n)$ by Bernstein and Nanongkai~\cite{bernstein2019distributed}. Our algorithm achieves the same near-optimal time complexity as Bernstein and Nanongkai. But we think the new technique might be of interest. Our method is completely different and it solves APSP as $n$ independent SSSP computations. Furthermore, it can be seen as a step toward a deterministic APSP algorithm with time complexity $\tilde{O}(n)$, which remains open. The reason is as follows: the algorithm of Bernstein and Nanonkai uses randomness in three parts, (1) for sampling some center nodes along long paths, (2) for random filtering of broadcasts, ensuring that only small congestion passes through each edge, and (3) for random-delays scheduling. Our APSP algorithm uses randomness only for scheduling. Deterministic scheduling remains an important open problem in distributed graph algorithms. Any deterministic scheduling with bounds with $\poly(\log n)$ factor of the random delays method would make our APSP algorithm deterministic and $\tilde{O}(n)$ round complexity, thus resolving a classic open problem~\cite{bernstein2019distributed, agarwal2018deterministic,agarwal2020faster}. The best known deterministic distributed exact APSP is due to Agarwal and Ramachandran~\cite{agarwal2020faster} and runs in $\tilde{O}(n^{4/3})$ rounds, and that was an improvement on an earlier work of Agarwal, Ramachandran, King, and Pontecorvi~\cite{agarwal2018deterministic} which had an $\tilde{O}(n^{3/2})$ round complexity.

\paragraph{Our approach to distributed SSSP.}
As mentioned above, the efficiency of Dijkstra's (sequential) algorithm is rooted, in part, in that it relaxes each edge only a single time. However, for that, the algorithm needs to repeatedly identify the minimum distance unvisited node in the entire network. The latter makes the algorithm inefficient in the distributed setting. Our approach tries to achieve a similar effect of each edge getting relaxed only a small number of times---and concretely each edge communicating only $\poly(\log n)$ messages---without so much global coordination.

Our basic idea, being overly optimistic, is to set up a recursive algorithm: Let $W$ be the maximum edge weight and notice that $dist(s, v)\leq nW$ for all nodes $v\in V$. Let $d=nW$. Suppose that we had an exact \textit{cutter algorithm} that could distinguish the set $V_1$ of nodes $v$ such that $dist(s, v)\leq d/2$ from the rest of nodes $V_2=V\setminus V_1$. Then, we would first remove $V_2$ from the graph and solve SSSP recursively only among the vertices of $V_1$ with a maximum distance upper bound of $d/2$. Once that finishes, we would solve SSSP among only the vertices of $V_2$ (by removing $V_1$ nodes but making each $u\in V_2$ node connected to $v\in V_1$ simulate a source connected only to $u$ and at distance $dist(s, v)+w(uv)$ from $u$). There is some subtlety in this, in how we would inform $V_2$ nodes to start, but let us ignore that for now. If we manage to make this scheme work, each node would be active in only one side of the two-way split recursion. Thus, each edge would be involved in $O(\log n)$ cutter algorithms. Given a $\poly(\log n)$-congestion exact cutter, the overall algorithm would have a $\poly(\log n)$-congestion. But where do we get such an exact cutter? 

We use instead an approximate cutter. Nanongkai's SSSP approximation algorithm~\cite{nanongkai2014distributed} can give us a cutter with a small additive error. Concretely, it can identify a set $V_1$ of vertices with the following guarantees: (2) For all nodes $v$ such that $dist(s, v)\leq d$, we have $s\in V_1$, and (2) for all nodes $v\in V_1$, we have $dist(s, v)\leq d/2+\eps d/2$ for desirably small constant $\eps\in (0, 0.1)$. The algorithm has $\tilde{O}(n)$ rounds and $\poly(\log n)$ congestion. Indeed, it consists of running one BFS in a graph with suitably rounded weights. We plug in this approximate cutter instead of our ideal exact cutter. We comment that, perhaps surprisingly, any $\poly(\log n)$ factor approximation would also be equally useful here.

Then, we move into the first half of recursion only with nodes of $V_1$, i.e., in a graph where nodes of $V\setminus V_1$ are removed. The recursion guarantees is that, by the end, all nodes $v$ that have $dist(s, v)\leq d/2$ will learn their exact distance $dist(s, v)$. Then, for the second half of recursion, we effectively remove all these nodes and continue with only the nodes for which $dist(s, v)>d/2$. Notice that we lose the nice property that each node goes only in one of the two splits of the recursion. However, one can see that, in the recursion tree, each node is active in $O(\log n)$ subproblems. Furthermore, throughout, we are working with only undirected graphs and computing undirected BFSs; this property is crucial in allowing us to extend the result to have a low energy complexity.

One extra challenge, which should be at least mentioned is as follows (we admit that the solution is hard to summarize cleanly and probably more understandable in the actual algorithm): How do nodes of the second half know when to start their recursion? Their synchrony is crucial for BFS computations. Using a global broadcast to announce/coordinate that start would necessitate $\Theta(D)$ time, where $D$ denotes the network diameter. Considering that the recursion tree has at least $\Omega(n)$ subproblems, this would make the time complexity as slow as $\Omega(nD)$. We use instead a maximal forest of the set of nodes relevant to the subproblem at hand and coordinate through that. Because of this, different subproblems at the same part of the recursion might proceed at different speeds, and we need some care to show that this does not cause a problem.   

Finally, because of the recursion, we will end up having to solve the more general closest-source shortest path problem (CSSP), where for given set $S$ of sources, we want each node $v$ to know $dist(S, v)=min_{s\in S} dist(s, v)$. This requires coordination among the processes starting from different sources, and creates algorithmic complications. We do not discuss those issue here.

\subsection{With Energy Considerations}
\label{subsec:WithEnergy}
We next review the definition of energy complexity. We note that this measure was studied long ago primarily in the context of wireless networks---with an initial motivation coming from battery-powered devices, see e.g. \cite{jurdzinski2002efficient, jurdzinski2002energy}. However, over the past few years, there has been growing interest in understanding the energy complexity of various network computations in the basic message-passing model of distributed computing. See e.g., \cite{chang2017exponential, chang2018energy, chang2020energy, ChaterjeeGP20, DBLP:conf/wdag/BarenboimM21, DBLP:conf/podc/AugustineMP22, ghaffari2022average, dufoulon2023distributed, ghaffari2023distributed}. \footnote{To the best of our knowledge, our paper is the first to show that, distributed algorithms with low energy can be useful also for more classic problems in \congest/\local models (which do not directly focus on energy), because the low energy bound implies low congestion and allows us to run several algorithms essentially simultaneously.}

\paragraph{Model and Energy Complexity.} We continue with the standard synchronous message-passing model of distributed computing, but with the additional property that each node can choose to sleep for some rounds. In a round that a node is sleeping, it performs no computations and cannot send or receive any messages. In particular, any messages sent to it in this round are lost. This variant has sometimes been called the \textit{sleeping} model\cite{DBLP:conf/podc/AugustineMP22, dufoulon2023distributed}. The underlying modeling assumption is that, a node consumed negligible energy during a sleeping round. Hence, the number of rounds that the node is awake gives an (asymptotic) measure of the energy that it spends. The \textit{energy complexity} of an algorithm is the maximum energy spent by a node, i.e., the maximum over all nodes of the number of rounds in which this node is awake.

It is worth noting that energy complexity directly implies a bound on congestion: if each node is awake only $T$ rounds, each edge can have at most $T$ messages sent through it in each direction. As such, studying energy-efficient algorithms also leads to low-congestion algorithms. Generally, low-congestion algorithms allow more concurrent schedulings, and this can be a strong positive, either when for an algorithmic problem we need to run several instances of this algorithm simultaneously (e.g., in the APSP case discussed before), or when the computer network naturally runs several distributed algorithms/communication concurrently. Of course, a small congestion does not necessarily imply that a node is involved in a small amount of communication.

\paragraph{Computing SSSP/BFS with Low-Energy.} With the $\poly(\log n)$ energy complexity constraint, SSSP, and even BFS, become challenging. The difficulty, even in BFS, is that a node does not know its distance from the source. Thus, it does not know when it should be awake, listening to receive the message of the arriving BFS. Recall that when a node is sleeping, messages sent to it are lost. 

In a recent work, Chang, Dani, Hayes, and Pettie~\cite{chang2020energy} gave the first BFS algorithm with somewhat small energy and near-optimal time: their algorithm is randomized and computes BFS in ${D} \cdot 2^{O(\sqrt{\log n \log\log n}) }$ time and using $2^{O(\sqrt{\log n \log\log n})}$ energy. A reparametrized version would have $\poly(\log n)$ energy and $D \cdot n^{\eps}$ time complexity for an arbitrarily small constant $\eps>0$. This work was presented in the radio network model, which has an additional constraint on simultaneous messages, but for the BFS problem that was not a major obstacle because of the known decay protocol that manages such simultaneous transmissions and delivers at least one message to the node. This solution was the state of the art even with the basic message-passing model. In a follow-up work, Dani and Hayes~\cite{dani2022wake} essentially resolved the randomized version of the BFS problem by giving a randomized distributed algorithm with energy complexity $\poly(\log n)$ and time complexity $\tilde{O}(D)$. Their solution depended on some nice properties of randomized low-diameter graph decomposition~\cite{Miller2013ParallelShifts}, which are not known for deterministic decompositions. It remained open whether one can obtain similar bounds using a deterministic algorithm. Furthermore, this was just the unweighted version of the SSSP problem. It remained completely open whether one can obtain an algorithm with near-optimal time and energy complexity for the general case of SSSP. 

\paragraph{Our results on Energy-Efficient SSSP.} We give a deterministic algorithm that computes the exact SSSPs with $\tilde{O}(n)$ time complexity and $\poly(\log n)$ energy complexity. As is standard, in the presentation we assume that each edge $e$ has a weight $w(e)\in \{1, 2, \dots, \poly(n) \}$. More generally, if all weights are in $\{1, \dots, W\}$, the result generalizes with $\log (nW)$ factors in the bound. As a simpler special case of our SSSP, and in fact as a subroutine for the general case, we also provide an SSSP algorithm for BFS (i.e., SSSP in unweighted graphs) with $\tilde{O}(D)$ time complexity and $\poly(\log n)$ energy complexity. Here, $D$ denotes the hop diameter of the network.  
\begin{theorem}
There is a deterministic distributed algorithm that computes exact SSSPs with $\tilde{O}(n)$ time complexity and $\poly(\log n)$ energy complexity. Its special case for unweighted graphs computes an exact BFS with $\tilde{O}(D)$ time complexity and $\poly(\log n)$ energy complexity.    
\end{theorem}

\paragraph{A few words about the method.} Our solution for SSSP relies directly on the first ingredient that we outlined in the previous subsection, the distributified variant of Dijkstra's algorithm. But that itself needs a solution for the BFS problem, and to have a deterministic algorithm, we need an energy-efficient deterministic algorithm for BFS. This is the second technical novelty in our paper. 

Our solution uses an idea similar to the approach of Dani and Hayes~\cite{dani2022wake}: we also use sparse neighborhood covers to coordinate the waking and sleeping of the nodes during the growth of BFS. Their algorithm relies on certain nice properties of the celebrated Miller, Peng, Xu~\cite{Miller2013ParallelShifts} randomized constructions of sparse neighborhood cover, and we are not aware of a deterministic construction with similar properties. We give a different scheme based on the deterministic sparse neighborhood cover construction of Rozhon and Ghaffari~\cite{Rozhon2020Polylogarithmic-timeDerandomization}, which basically needs only an energy-efficient subroutine for computing a BFS up to a certain distance. We break this seemingly vicious circle of dependencies between sparse neighborhood cover and BFS by setting up an appropriate recursion based on the distance traveled by BFS, and use some other ideas to stitch the solutions together. We hope our deterministic low-energy constructions of neighborhood covers might find applications in deterministic energy-efficient algorithms for a wider range of problems. 

Finally, we emphasize that in our SSSP approach, it is critical that the BFSs that need to be computed are in undirected graphs.\footnote{Some exact SSSP approaches use \textit{scaling} to work internally with approximations and boost them to exact distances---see, e.g., \cite{ghaffari2018improved,forster2018faster,chechik2020single,rozhovn2023parallel}. Unfortunately, all of these end up having to compute distances in certain directed graphs (more precisely, graphs with asymmetric weights along the two edge directions). This creates a major obstacle for extensions to energy-efficient computations, which rely on properties of neighborhood covers in undirected graphs.} This is because our energy-efficient BFS solution relies crucially on sparse neighborhood covers in undirected graphs, and it is not clear how to obtain similar results in directed graphs (indeed, it is not even clear what is the neighborhood cover concept for directed graphs, which would have the right properties). 




\section{Closest-Source Shortest Paths in $\tilde{O}(n)$ Time and $\tilde{O}(1)$ Congestion}
In this section, we discuss an algorithm that, given a set $S$ of sources, computes the distance $dist(S, v)=min_{s\in S} dist(s, v)$ for each node $v$, in $\tilde{O}(n)$ time and using $\tilde{O}(1)$ congestion per edge.

\subsection{Premilinaries}
\label{sec:prelim}
We first review two subroutines from prior work: (1) a rounding approach that provides $(1+\eps)$ approximation of shortest paths, which we review in \Cref{subsubsec:mssp-approx}, and (2) a distributed algorithm for computing a maximal spanning forest, which we review in \Cref{subsubsec:msf}.

\subsubsection{CSSP Approximation}
\label{subsubsec:mssp-approx}
Let us denote the set of sources by $S$. Also, for now, we assume that the weights of all edges are positive integers in $[1, \poly(n)]$. We will come back to the case of edges with weight $0$ later in \cref{theorem:main_zero}. 

If the maximum weighted distance from the sources is $O(n)$, we can run BFS from $S$, waiting for $t$ rounds for the edge with weight $t$. We cannot afford to do this if the weighted diameter is large. But we can use a rounding trick, first used in the distributed setting by Nanongkai~\cite{nanongkai2014distributed}, to approximate the distances from all nodes to $S$. The following statement abstracts the resulting algorithm. \fullOnly{A proof is provided in \Cref{app:prelim}, for self-containedness.}\shortOnly{A proof is provided in the full version of this paper, for self-containedness.} 

\begin{lemma}
    \label[lemma]{lemma:approx}
Consider a graph $G = (V, E)$ and a set of sources $S$. There is an algorithm that, given $\epsilon \in (0, 1)$ and an integer $W>0$, runs in time $O(\frac{n}{\epsilon})$ and congestion $O(1)$, and for each node $v$ outputs $dist'(S, v)$ with the following guarantees:

    \begin{itemize}
        \item If $dist'(S, v) \neq \infty$, then $dist(S, v) \le dist'(S, v) < dist(S, v) + \epsilon W$
        
        \item If $dist'(S, v) = \infty$, then $dist(S, v) > 2W$
    \end{itemize}
\end{lemma}

\subsubsection{Maximal Spanning Forest}
\label{subsubsec:msf}
In our algorithm, we make use of a classic and well-known distributed algorithm for computing a maximal spanning forest~\cite{boruuvka1926jistem, gallager1983distributed}, which runs in near-linear time and with polylogarithmic congestion. The following statement summarizes this result:

\begin{theorem}
    \label[theorem]{theorem:spanning}[Boruvka's Algorithm\cite{boruuvka1926jistem, gallager1983distributed}]
There is a deterministic distributed algorithm that, on any undirected graph with $n$ nodes, runs in time $O(n\log{n})$ and congestion $\poly(\log{n})$ and computes a maximal spanning forest of the graph.
\end{theorem}

\subsection{An Algorithm Idea for CSSP, Assuming an \textit{Exact Cutter}}
\label{subsec:exactCutter}

Once again, we are given a weighted graph $G$ and a set of sources $S$. Imagine we had a way to solve the following problem: given $\mathcal{D}$, determine all nodes $v \in V$ with $dist(S, v) \le \mathcal{D}$. By this, we mean that after running this algorithm, every node should know whether it is at a distance at most $\mathcal{D}$ from the sources or not. Let us refer to such an algorithm as to \textbf{exact cutter}.

If we had an exact cutter, we could approach CSSP as follows: Let $\mathcal{D} = n\cdot \max {w_{e}} \leq \poly(n)$. This is an upper bound on the maximum distance from $S$ to other nodes. We could solve the CSSP problem recursively, as follows. Let $\mathcal{D}_1 = \frac{\mathcal{D}}{2}$. 

\begin{enumerate}
    \item Use the exact cutter to determine all nodes at the distance at most $\mathcal{D}_1$ from $S$. Let us denote this set of nodes as $V_1$.
    \item Solve the CSSP problem for $V_1$, recursively, by removing all nodes in $V\setminus V_1$.
    \item For every edge $(v, u)$ with $v \in V_1, u \in V \setminus V_1$, create an imaginary node $x_{vu}$ somewhere on the edge $(u, v)$, splitting it into two edges $(v, x_{vu})$ and $(x_{vu}, u)$, so that $w((v, x_{vu})) = \mathcal{D}_1 - dist(S, v)$. Let us denote this set of imaginary nodes as $X$. These nodes from $X$ form a ``cut'' at distance $\mathcal{D}_1$ from $S$. 
    \item Finally, solve the problem recursively for the set $X\cup (V \setminus V_1)$, with nodes $X$ as the new sources (and where we have removed nodes of $V_1$ from the problem). Notice that we do not really need the imaginary nodes to be physically present: For every imaginary node $x_{vu}$, node $u \in (V\setminus V_1)$ can simulate the messages of $x_{vu}$ (this node can send messages only to $u$). 
\end{enumerate}
The positive of this approach is that we trivially break the problem into two similar problems, solved sequentially, each with a $2$-factor smaller maximum distance, and each edge is active in only one of the two branches of the recursion. The only real cost of the recursion is the invocation of the cutter. One can see that each edge would be involved in $O(\log n)$ instances of the exact cutter. 

Of course, unfortunately, we do not have an exact cutter. But, this idea helps us to develop the actual algorithm. In the next subsection, we discuss our approach of using an approximate cutter, and how to manage the complexities arising from the inexactness of this cutter.

\subsection{The Algorithm for CSSP, using an Approximate Cutter}

\label{subsec:mainAlgo}
We do not have an exact cutter and we do not know how to determine exactly the set $V_1$ of nodes of distance at most $\mathcal{D}$ from $S$. But, we can use the algorithm from \cref{lemma:approx} to approximate the distances, and that gives us an overestimate of $V_1$ by including all nodes whose distance can be at most $\mathcal{D}$ given the approximate distance that we have. Notice that this set might be much larger than $V_1$. Furthermore, if we use this approximation, we lose the nice property that each node is involved in only one half of the recursion, because some nodes included in this overestimated $V_1$ will turn out to have a distance greater than $\mathcal{D}$ and thus will have to participate in the second half of the distance computation recursions as well. But we argue next that we can remedy these issues. 




\begin{definition} For any nonnegative integer $\tau$, we define the  \textbf{$\tau$-thresholded CSSP} problem as follows: Let $S$ be the set of source nodes. For each node $v$, if $dist(v, S) \le \tau$, then $v$ should output $dist(v, S)$. If $dist(v, S) > \tau$, then $v$ should output a special symbol $\infty$. 
\end{definition}


To compute CSSP we simply run a $\mathcal{D}$-thresholded CSSP for $\mathcal{D}=O(n\cdot maxW)$. Next, we describe how we perform the $\mathcal{D}$-thesholded CSSP Algorithm recursively, using a structure similar to the one described in \cref{subsec:exactCutter}, but via an approximate cutter.

\paragraph{The $\mathcal{D}$-thesholded CSSP Algorithm.} The algorithm is recursive:

\begin{enumerate}

    \item If $\mathcal{D} = 1$, we are in the base case. In this case, the only nodes with $dist(S, v) \le \mathcal{D}$ are the sources themselves and nodes that are connected to some source by an edge of weight $1$. All nodes can detect this in one round. If $\mathcal{D}>1$, proceed to the next steps.

    \item Compute all connected components of $G$, and a spanning tree for each of them, with an algorithm from \cref{theorem:spanning}. We solve this problem for each component independently.
    
    \item Choose $\epsilon = 0.5$, and compute an approximation $dist'(S, v)$ of distances from $S$ to each node $v$, using the algorithm from \cref{lemma:approx}. Let $V_1$ denote the set of nodes $v$ with $dist'(S, v) <\mathcal{D} + \epsilon \mathcal{D}$. For any $v \in V_1$, we have $dist(S, v) < \mathcal{D} + \epsilon \mathcal{D}$. Furthermore, for any $v \in V$ with $dist(S, v) \le \mathcal{D}$, we have $dist'(S, v) < \mathcal{D} + \epsilon \mathcal{D}$, so $v \in V_1$.
    
    \item Let $\mathcal{D}_1 = \frac{\mathcal{D}}{2}$. Perform a $\mathcal{D}_1$-thresholded CSSP for nodes $V_1$, with set $S$ of sources. Nodes from $V\setminus V_1$ do not participate in this recursion call. In every connected component $C$, collect with convergecast via the spanning tree of $C$, whether all nodes of $C$ which are in $V_1$ are done with this $\mathcal{D}_1$-thresholded CSSP (in particular, if $V_1$ is empty, no recursion is called, and already all nodes of $C\cap V_1$ are done with the $\mathcal{D}_1$-thresholded CSSP). Here by convergecast we mean that every node will tell its parent when it and its entire subtree is done with $\mathcal{D}_1$-thresholded CSSP.
    When the root of the spanning tree of $C$ detects that all nodes of $C$ are done with this $\mathcal{D}_1$-thresholded CSSP step, it chooses a start time for the final step of the algorithm. It sets it to be in $\Theta(|C|)$ rounds into the future and then spreads this starting round number via the spanning tree of $C$.
    
    \item 
    Let $V_2$ denote the set of all nodes $v$ with $d(S, v) \le \mathcal{D}_1$. Notice that after previous steps, all nodes know whether they are in $V_2$ or not. For every edge $(v, u)$ with $v \in V_2, u \in V_1 \setminus V_2$, create an imaginary node $x_{vu}$ somewhere on the edge $(u, v)$, splitting it into two edges $(v, x_{vu})$ and $(x_{vu}, u)$, so that $w((v, x_{vu})) = \mathcal{D}_1 - dist(S, v)$. Let us denote this set of imaginary nodes as $X$. These nodes from $X$ form a ``cut'' at distance $\mathcal{D}_1$ from $S$. 

    \item Finally, perform $\mathcal{D}_1$-thresholded CSSP for nodes from $X\cup (V_1 \setminus V_2)$ with set $X$ of sources. Nodes from $V_2$ do not participate in this recursion call. For every imaginary node $x_{vu}$, node $u \in V_1 \setminus V_2$ can simulate the messages of $x_{vu}$ (it only ever sends messages to $u$). 

    For any node $v \in V \setminus V_2$, we know that $dist(S, v) = dist(X, v) + D_1$. Since $\mathcal{D} = 2\mathcal{D}_1$, if $dist(X, v) \le \mathcal{D}_1$, we have $dist(S, v) = \mathcal{D}_1 + dist(X, v) \le \mathcal{D}$, and if $dist(X, v) > \mathcal{D}_1$, we have $dist(S, v) = \mathcal{D}_1 + dist(X, v) > \mathcal{D}$. Hence, the algorithm computes the correct output.
\end{enumerate}

\subsection{Analysis of the Algorithm}

\label{subsec:analysis}

Let us denote the running time of $\mathcal{D}$-thresholded BFS for a set of nodes $V$ and a set of sources $S$ by $T(V, S, \mathcal{D})$. Then, the recurrence from \cref{subsec:mainAlgo}, keeping all the notations from that section, would have the following form:

\begin{align*}
T(V, S, \mathcal{D}) = &\underbrace{O(|V|\log{|V|})}_\text{Connected components} + \underbrace{O(|V|)}_\text{Approximation} +  \underbrace{T(V_1, S, \mathcal{D}_1)}_\text{Call on $V_1$} \\ \\&+  \underbrace{O(|V|)}_\text{Convergecast} +  \underbrace{T(X \cup (V_1 \setminus V_2), X, \mathcal{D}_1)}_\text{Call on $V_1 \setminus V_2$}    
\end{align*}

Remember that all nodes in $X$ are simulated by the nodes in $V_1 \setminus V_2$, so let us replace $X \cup (V_1 \setminus V_2)$ with $V_1 \setminus V_2$ in the last term, and simplify:

$$T(V, S, \mathcal{D}) = O(|V|\log{|V|}) + T(V_1, S, \mathcal{D}_1) + T(V_1 \setminus V_2, X, \mathcal{D}_1)$$

\begin{lemma}
    \label[lemma]{lemma:appear}
    Every $v \in V$ appears in $V'$ for only $O(\log{\mathcal{D}})$ subproblems $T(V', S', \mathcal{D}')$.
\end{lemma}

\begin{proof}
    Let us look more carefully at our recursion. Subproblem $T(V', S', \mathcal{D}')$ actually denotes doing BFS from some distance $k\mathcal{D}'$ from the original sources $S$ up to distance $\mathcal{D}'$, which recurses to doing BFS from distance $k\mathcal{D}'$ up to distance $\frac{\mathcal{D}'}{2}$, and then from distance $k\mathcal{D}' + \frac{\mathcal{D}'}{2}$ up to distance $\frac{\mathcal{D}'}{2}$ again. Also, note that at every level except the very highest, all nodes in $V'$ have a distance from $S$ in range $[k\mathcal{D}', k\mathcal{D}' + 2\mathcal{D}' + \epsilon \cdot 2\mathcal{D}') = [k\mathcal{D}', k\mathcal{D}' + 3\mathcal{D}')$: when we recurse to $T(V_1, S, \mathcal{D}_1)$, we only have nodes at distance $\le \mathcal{D} + \epsilon \mathcal{D} = 2\mathcal{D}' + \epsilon \cdot 2\mathcal{D}'$ from the sources.

    It follows that every node can appear in at most $3 = O(1)$ recursive subproblems per every level $\mathcal{D}'$. So, every node appears in $O(\log{\mathcal{D}})$ subproblems in total.
\end{proof}

\begin{corollary}
    \label[corollary]{corollary:appearSum}
    The total sum of $|V'|$ over all recursive subproblems $T(V', S', \mathcal{D}')$ is $O(|V| \log{\mathcal{D}})$.
\end{corollary}

\begin{proof}
    Since by \cref{lemma:appear} each node appears in $O(\log{D})$ subproblems, the total sum of $|V'|$ over all subproblems is $O(|V|\log{D})$.    
\end{proof} 

\begin{theorem}
    \label[theorem]{theorem:main_no_zero}
    Consider a graph $G = (V, E)$ and a set of sources $S$. If the weights of all edges in $G$ are positive integers bounded by $O(\poly(n))$, there exists an algorithm that computes all distances $dist(S, v)$, in time $O(n\log^2{n})$ and congestion $O(\log^2{n})$.
\end{theorem}

\begin{proof}
We run $T(V, S, \mathcal{D})$ with $\mathcal{D} = 2^L$, where $L$ is the smallest integer such that $2^L \ge n\cdot maxW$ is an upper bound on the distances from $S$ to other nodes. Since all weights are $O(\poly(n))$, we have $O(\log{\mathcal{D}}) = O(\log{n})$. Next we analyze the recursion: by \cref{corollary:appearSum}, the sum of sizes of $|V'|$ over all subproblems is $O(|V|\log{\mathcal{D}})$. So, the sum of $O(|V'|\log{|V'|})$ is $O(|V|\log{\mathcal{D}}\log{|V|})$. Therefore, $T(V, S, \mathcal{D}) = O(|V|\log{\mathcal{D}}\log{|V|})$. Now we analyze the congestion. Each edge is involved in $O(\log{\mathcal{D}})$ subproblems, and in $O(\log{\mathcal{D}})$ computations of connected components and their trees. From \cref{lemma:approx} and \cref{theorem:spanning}, it follows that the congestion is $O(\log{n} \log{\mathcal{D}}) = O(\log^2{n})$.
\end{proof}

\paragraph{Extension to the case of edges with zero weight.} The following statement extends the results to the case with $0$ weight edges. \fullOnly{Since this is a standard idea, we defer the proof to \Cref{app:zero-weight}.}\shortOnly{Since this is a standard idea, we defer the proof to the full version.}
\label{subsec:zeroWeight}

\begin{theorem}
    \label[theorem]{theorem:main_zero}
    Consider a graph $G = (V, E)$ and a set of sources $S$. If the weights of all edges in $G$ are nonnegative integers bounded by $O(\poly(n))$, \textbf{and can be zero}, there exists an algorithm that computes all distances $dist(S, v)$, in time $O(n\log^2{n})$ and congestion $O(\log^2{n})$.
\end{theorem}

\section{Closest-Source Shortest Paths in $\tilde{O}(n)$ Time and $\tilde{O}(1)$ Energy}
In this section, we discuss an algorithm for the closest-source shortest paths problem. Given a set $S$ of multiple sources, the algorithm computes the distance $dist(S, v)=min_{s\in S} dist(s, v)$ for each node $v$. It uses $\tilde{O}(n)$ time, and $\tilde{O}(1)$ energy per node. In the first six subsections of this section, we consider only unweighted graphs (so that all edges have weight $1$). In \cref{subsec:BFSwithcovers}, we develop a way to perform a $\mathcal{D}$-thresholded BFS, given a layered sparse $\mathcal{D}$-cover (defined in \cref{subsubsec:forest}). In \cref{subsec:energybfs}, we explain how to perform a multi-source BFS even without being given the sparse covers. Finally, in \cref{subsec:energymsps} we extend this algorithm to CSSP in weighted graphs.

\subsection{Preliminaries}

\subsubsection{Collecting Information in a Tree in the \energy Model}

\label{subsubsec:collecting}
Assume that there is a rooted tree of depth $d$, where all nodes know their children, parents, and depth. Let us denote the depth of a node $u$ by $depth(u)$. Let us say that these nodes want to collect some information via convergecast, and then broadcast it down to all the nodes. How would they do this with a small energy cost?  

We will utilize the assumption that all nodes know their depth in the tree. Let us define a \textbf{period} $p$ of the tree, and two procedures: 

\textbf{Convergecast.} We can propagate any information up as follows: node $v$ will wake up at rounds $p - depth(v) - 1, p - depth(v), 2p - depth(v) - 1, 2p - depth(v), \ldots$, and, if it has any information to propagate, it will send it to the parent. 

\textbf{Broadcast.} We can propagate any information down as follows: node $v$ will wake up at rounds $depth(v), depth(v) + 1, p + depth(v), p + depth(v) + 1, 2p + depth(v), 2p + depth(v), \ldots$, and, if it has any information to propagate, it will send it to all its children. 

Let us call a node \textbf{active} for a cluster $C$, if it is involved in both convergecast and broadcast for this cluster. Then, by the time $t$ any node $v$ received any signal, and all the nodes of $C$ are active, then all nodes of $C$ will receive this signal by the round $t + O(d + p)$. Additionally, every node wakes up only in $\Theta(\frac{1}{p})$ fraction of rounds.

\subsubsection{Rooted Spanning Tree with Low Energy}
An adaptation of the classic Boruvka approach to the minimum spanning tree~\cite{boruuvka1926jistem, gallager1983distributed} gives an $\tilde{O}(n)$-time and low-energy distributed algorithm for computing a maximal forest~\cite{augustine2022MSTAwake}.
\begin{theorem}
    \label[theorem]{theorem:spanningenergy}[Boruvka\cite{boruuvka1926jistem}, adapted by Augustine et al.~\cite{augustine2022MSTAwake}]
    There is a deterministic distributed algorithm that, given any $n$-node graph, computes a maximal forest of it using $\tilde{O}(n)$ time and $\poly(\log n)$ energy. In particular, $O(n\log^2{n})$ time and $O(\log^2{n})$ energy is sufficient.
\end{theorem}

\subsubsection{Running Several Subroutines Simultaneously}
\label{subsubsec:running}
In our algorithms, we often have several algorithmic subroutines. A question arises when these subroutines want to use the same edge. We use the following approach: if an edge is going to be used by at most $k$ subroutines at every single round, we divide all rounds into \textbf{megarounds}, each consisting of $k$ rounds. Then, if a node has to be awake in any of these rounds, it will be awake in all of them. Then, in this single megaround, we will be able to deliver all messages that subroutines wanted to deliver in the corresponding round.

\subsection{Sparse Covers}

\subsubsection{Definition}
\label{subsubsec:sparseCovers}
\begin{definition} For any positive integer $d$, we define a \textbf{sparse $d$-cover} with stretch $s$ of an $n$-node graph  to be a set of clusters, such that the following conditions hold:

\begin{itemize}
    \item The diameter of each cluster is at most $d\cdot s$.

    \item Each node is in $O(\log{n})$ clusters.

    \item For each node $v$, there  is a cluster that contains all nodes $u$ for which $dist(u, v)\leq d$.
\end{itemize}
We assume that each node knows which clusters it is in and for each cluster, we have a tree of depth $d\cdot s$ that spans all cluster nodes, where each node knows its parent, children, and depth in the tree.

\end{definition}

There exist sparse covers with stretch $s=O(\log n)$~\cite{Awerbuch1990SparsePartitions, Linial1993LowDecompositions} but the constructions are not deterministic or efficient. We work with deterministic distributed constructions of sparse cover that are time efficient and have slightly higher parameters~\cite{Rozhon2020Polylogarithmic-timeDerandomization}, as abstracted below:

\begin{itemize}
    \item The cluster stretch is $O(\log^3{n})$. More concretely, each cluster has a $O(d\log^3{n})$-depth cluster tree, which spans all cluster nodes.
    \item Each edge is used in only $O(\log^4{n})$ cluster trees.
\end{itemize}

\subsubsection{Forest of Sparse Covers}

\label{subsubsec:forest}

Observe the following: 

\begin{observation}

\label[observation]{observation:clusterparent}

Let $\mathcal{C}_1$ be a sparse $d$-cover, and $\mathcal{C}_2$ be a sparse $2ds$-cover, where $s$ is larger than the stretch of $\mathcal{C}_1$ (so that $s = \Theta(\log^3{n})$). Then, every cluster of $\mathcal{C}_1$ is fully contained in some cluster of $\mathcal{C}_2$, together with its $ds$-neighbourhood.

\end{observation}

\begin{proof}

Consider any cluster $C$ of $\mathcal{C}_1$, and any node $v$ of $C$. All its $2ds$-neighbourhood is contained in some cluster of $\mathcal{C}_2$. As $diam(C)\leq ds$, $C$ and its $ds$ neighborhood are contained in this cluster.
\end{proof}

This allows us to build a tree-like structure on clusters of sparse covers of different layers. Choose $B = \Theta(\log^3{n})$ so that $\frac{B}{2}$ is an upper bound on the stretch of sparse covers.

\begin{definition}

For any $\mathcal{D}$, we define a \textbf{layered sparse $\mathcal{D}$-cover} to be a collection of sparse $B^{j}$-covers for all $j\in \{0,1, \dots, \lceil{\log_B 2\mathcal{D}\rceil}\}$, with the following constraint:

\begin{itemize}
\item
For any $j < \lceil{\log_B 2\mathcal{D}\rceil}$, every cluster $C$ of sparse $B^j$-cover is assigned a cluster of sparse $B^{j+1}$-cover that completely contains it and its $\frac{B^{j+1}}{2}$-neighborhood (nodes at distance at most $\frac{B^{j+1}}{2}$ from $C$) (every node of $C$ has to know the identifier of this cluster). Such cluster is called a \textbf{parent cluster} of $C$, and is denoted by $parent(C)$.
\end{itemize}

\end{definition}

\subsection{$\mathcal{D}$-Thresholded BFS given a Layered Sparse $\mathcal{D}$-Cover}

\label{subsec:BFSwithcovers}

Suppose we are already given a layered sparse $\mathcal{D}$-cover of the graph. We will see how to construct this in \cref{subsec:energybfs}. Let $L = \lceil{\log_B \mathcal{D}\rceil}$, and $S$ be the set of sources of BFS.

\textbf{Idea of the algorithm.} We will have two processes in parallel: BFS, and collecting information in the cluster trees of sparse covers, as described in \cref{subsubsec:collecting}. For a node $v$, we want to ensure that it is active when BFS actually reaches it. For that, we want the frequency with which $v$ is awake to increase as BFS comes closer and closer to $v$. We will achieve this by having smaller clusters operate at higher frequencies: clusters of sparse $B^i$-covers will operate with period $B^i$, and will be activated only when BFS gets ``pretty close" to them. More precisely, a cluster $C$ will get activated when $parent(C)$ detects that the BFS got pretty close, and tells $C$ to activate. Now, let us formalize this idea by providing the exact details of the algorithm. 



First, we will introduce the notion of \textbf{relevance}.

\begin{definition}
    A cluster of the $B^L$-cover is called \textbf{relevant} if it contains nodes from $S$. Otherwise, let us call it \textbf{irrelevant}. Let us recursively extend the notion of relevancy as follows: a cluster $C$ of $B^j$-cover with $j<L$ is relevant if and only if $parent(C)$ is relevant.
\end{definition}

\begin{lemma}
    For any $j<L$, every cluster of the $B^j$-cover which contains a node $v$ with $dist(S, v) \le \mathcal{D}$ is relevant.
\end{lemma}

\begin{proof}
    Consider the $L-1-j$-th ancestor of such a cluster, denote it by $C$. $C$ contains some node $v$ with $dist(S, v) \le \mathcal{D}$. As $parent(C)$ contains $\frac{B^L}{2}$-neighborhood of $C$, and $\frac{B^L}{2} \ge \mathcal{D}$, $parent(C)$ contains at least one node from $S$, and therefore is relevant.
\end{proof}

The idea is that $\mathcal{D}$-thresholded BFS will never get to irrelevant clusters, so irrelevant clusters should not be involved in this activating/deactivating process. Now, we describe how clusters are activated, deactivated, initialized, and how do they learn whether they are relevant.

\textbf{Activating clusters.} When some nodes of cluster $C$ of sparse $B^j$-cover are active, they will collect the following information with the convergecast and broadcast: 

\begin{center}
\textit{"Are any of nodes of $C$ reached by BFS?"} 

\end{center}

If all nodes of $C$ are active, then in at most $O(B^j\log^3{n})$ rounds after some node of $C$ becomes reached by BFS, all its nodes will learn that the cluster has been reached by BFS. When a node of $C$ learns this, it activates itself in all children clusters of $C$ in which it is present.

\textbf{Deactivating clusters.} A cluster will deactivate after it has been reached by BFS and all its children clusters have been activated (which is learned by the same convergecast and broadcast process). Additionally, for clusters of $B^0$-cover we require that all nodes of a cluster are reached by BFS before it deactivates.

\textbf{Initialization.} We start by making one cycle of convergecast-broadcast for every single cluster, which takes $O(B^L\log^3{n})$ rounds. By the end of that cycle, for every cluster $C$, all its nodes will know whether $C$ contains any node from $S$. Clusters of $B^L$-cover will learn whether they are relevant. Only relevant clusters of $B^L$-cover will be active and tell their children to activate; irrelevant clusters are never involved. Nodes will also learn, for each cluster $C$ in which they are contained, whether they should be active for $C$ (by checking whether $parent(C)$ contains any nodes from $S$).

After this initialization period of $O(B^L \log^3{n})$ rounds, all nodes (that have to) start doing convergecast-broadcast, and, simultaneously, BFS starts. One iteration of BFS will happen every $\Theta(\log^3{n})$ stages (with constant determined from \cref{lemma:induction}).

\textbf{Correctness argument.}






    

\begin{lemma}

    \label[lemma]{lemma:induction}

    At any point in time after the initialization period the following statement holds: 

    \begin{itemize}
        \item For any $0 \le j \le L$ and any relevant cluster $C$ of the sparse $B^j$-cover, all its nodes are active for $C$ before BFS reaches any of its nodes.
    \end{itemize}
    
\end{lemma}

\begin{proof}
    We can prove this by induction by $j$, from larger $j$ to smaller. For clusters, $C$, such that $C$ or $parent(C)$ contain nodes from $S$, this is ensured in the initialization stage. As all clusters of sparse $B^L$-covers contain simply all nodes the statement is true for $j = L$. Suppose that it is true for $j+1$, Let us prove it for $j$. Consider any cluster $C$ of sparse $B^j$-cover, such that $parent(C)$ does not contain vertices from $S$. Consider the moment when BFS first reaches $parent(C)$. 
    
    On the one hand, after at most $O(B^{j+1}\log^3{n})$ rounds, all nodes of $parent(C)$ will learn that $parent(C)$ has been reached, so all nodes of $C$ will become active for $C$. On the other hand, it will take at least $\frac{B^{j+1}}{2}$ steps for BFS to reach $C$ after that, with one step happening every $\Theta(\log^3{n})$ rounds, so it won't reach $C$ before all its nodes become active for $C$ (where we choose an appropriate constant for this $\Theta(\log^3{n})$, so that this BFS is indeed slow enough).
\end{proof}

\begin{theorem}
    \label[theorem]{theorem:BFSwithcovers}
    Given a layered sparse $\mathcal{D}$-cover, there exists an algorithm that correctly computes BFS in time $O(\mathcal{D}\log^{11}{n})$ and energy cost $O(\log^{18}{n})$.
\end{theorem}

\begin{proof}
The case $j = 0$ of \cref{lemma:induction} implies that every node will be awake when it's reached by BFS, so the BFS will work properly. We next bound the time and energy complexities.

Let us start with time. After $O(B^L\log^3{n}) = O(\mathcal{D}\log^6{n})$ rounds of initialization, one iteration of BFS happens every $\Theta(\log^3{n})$ rounds, and there are $O(\mathcal{D})$ of them, so there are $O(\mathcal{D}\log^6{n})$ rounds in total. But here we might have to deliver several messages through an edge at a time, hence the scheduling question arises. Remember that every edge is used in only $O(\log^4{n})$ cluster trees per sparse cover, so every edge is used in only $O(\log^5{n})$ cluster trees in total. So, in terminology from \cref{subsubsec:running}, we need megarounds, each of which contains $O(\log^5{n})$ rounds. Hence, the resulting time complexity becomes $O(\mathcal{D}\log^{11}{n})$.

Now, let's consider energy cost. Consider any cluster $C$ of a sparse $B^j$-cover. It becomes active only after some node in $parent(C)$ is reached. As diameter of $parent(C)$ is at most $O(B^{j+1}\log^3{n})$, in at most $O(B^{j+1}\cdot \log^3{n})$ BFS iterations after that, all nodes of $parent(C)$, and, therefore, of $C$ too, will be done with BFS, so $C$ will deactivate. As there is one iteration of BFS per $\Theta(\log^3{n})$ rounds, a node can be active for $C$ for at most $O(B^{j+1}\log^6{n})$ rounds. As it is awake with period $B^j$, it will be awake for $C$ for at most $O(B\log^6{n}) = O(\log^9{n})$ rounds. 

Note, however, that a node is in $O(\log^4{n})$ clusters in total, and that it has to remain awake for the entirety of every single megaround, which takes $O(\log^5{n})$ regular rounds. So, every node is awake for at most $O(\log^{18}{n})$ rounds.
\end{proof}

\subsection{Synchronous Deterministic Construction of $d$-Cover} 

\label{subsec:deterministiccover}

So far, we have assumed that we are given sparse covers. In this subsection, we review the synchronous construction of Rozhon and Ghaffari~\cite{Rozhon2020Polylogarithmic-timeDerandomization}. In the next subsection, we explain how to adapt this construction to the \energy model.

We start by introducing a deterministic algorithm for constructing the $d$-cover in $O(d\cdot \polylog(n))$ rounds in a synchronous environment. For that, we will a network decomposition. 

\begin{definition}
    ($k$-separated Weak Diameter Network Decomposition) Given a graph $G=(V, E)$, we define a $(\mathcal{C}, \mathcal{D})$ $k$-separated weak diameter network decomposition to be a partition of $G$ into vertex-disjoint graphs $G_1$, $G_2$, \dots, $G_\mathcal{C}$ such that for each $i\in \{1, 2, \dots, \mathcal{C}\}$, we have the following property: the graph $G_i$ is made of a number of vertex-disjoint clusters $X_1$, $X_2$, \dots, $X_\ell$, so that:

    \begin{itemize}
        \item For any $X_a$ and any two vertices $v, u \in X_a$, $dist(u, v) \le D$ in graph $G$.
        
        \item For any two $X_a, X_b$ with $a\neq b$ and any $u \in X_a$, $v \in X_b$ holds $dist(u, v) > k$.
    \end{itemize}
\end{definition}

\subsubsection{Construction of $k$-Separated Network Decomposition in the \congest Model}
We first recall the notion of Steiner trees: A Steiner tree of a cluster is a tree with nodes labeled as \textit{terminal} and \textit{nonterminal}; the aim is to connect all terminals, possibly using some nonterminals.

\begin{restatable}{theorem}{MAINALGGK}
\label[theorem]{theorem:main_alg_gk}[Rozhon and Ghaffari\cite{Rozhon2020Polylogarithmic-timeDerandomization}]
There is an algorithm that, given a value $k$ known to all nodes, in $O(k \log^{10} n )$ communication rounds outputs a $k$-separated weak-diameter network decomposition of $G$ with $O(\log n)$ color classes, each one with $O(k \cdot \log^3 n)$ weak-diameter in $G$. 


Moreover, for each color and each cluster $\mathcal{C}$ of vertices with this color, we have a Steiner tree $T_\mathcal{C}$ with radius $O(k \cdot \log^3 n)$ in $G$, for which the set of terminal nodes is equal to $\mathcal{C}$. 
Furthermore, each edge in $G$ is in $O(\log^4{n})$ of these Steiner trees.
\end{restatable}

\fullOnly{To remain self-contained, we provide a (simplified) description of the algorithm of \cite{Rozhon2020Polylogarithmic-timeDerandomization} in the \cref{app:main_alg_gk}.}\shortOnly{For accessibility, we provide a (simplified) description of the algorithm of \cite{Rozhon2020Polylogarithmic-timeDerandomization} in the full version of this paper.} Our low-energy construction will refer to the language used in this description. For proof of the correctness of their algorithm, we refer to \cite{Rozhon2020Polylogarithmic-timeDerandomization}.

\subsubsection{Construction of $d$-Cover in the \congest Model}

\begin{theorem}
    \label[theorem]{theorem:synchcover}
     There is an algorithm that, given a value $d$ that is known to all nodes, in $O(d \log^{10}{n})$ communication rounds outputs a sparse $d$-cover of this graph, together with the Steiner trees of its clusters, such that each edge appears in only $O(\log^4{n})$ cluster trees.
\end{theorem}

\begin{proof}
    We start by constructing a $2d+1$ separated weak network decomposition in $O(d \log^{10}{n})$ rounds, according to \cref{theorem:main_alg_gk}. Then, for each color separately, we expand each of its clusters to its $d$-neighborhood. As different clusters of the same color are at least $2d+1$ apart, the clusters remain disjoint, so each node will be in $O(\log{n})$ clusters (at most one per color). Additionally, each edge will join at most one cluster tree, so each edge still appears in only $O(\log^4{n})$ cluster trees.
\end{proof}

\subsection{Construction of Sparse $B^j$-Cover, given a Layered Sparse $B^{j-1}$-Cover in the \energy Model}

\label{subsec:energybfs}

The algorithm described in \cref{subsec:deterministiccover} for constructing sparse $d$-covers consists of simple multi-sources BFSs up to distance $O(d)$. It turns out that, given a layered sparse $B^{j-1}$-cover, we can construct a sparse $B^j$-cover with $\tilde{O}(1)$ energy cost in the \energy model, by making a small adjustment to our approach from \cref{subsec:BFSwithcovers}. 

\begin{theorem}
    \label[theorem]{theorem:coverwithcovers}
    Given a layered sparse $B^{j-1}$-cover, there exists an algorithm that works in time $O(B^j\log^{15}{n})$ and energy cost $O(\log^{25}{n})$, and computes 

    \begin{itemize}

    \item a sparse $B^{j}$-cover, and 

    \item for every cluster $C$ of sparse $B^{j-1}$-cover, assigns it a cluster of a newly found sparse $B^j$-cover that completely contains it and its $\frac{B^{j}}{2}$-neighborhood (so that every node of $C$ knows the identifier of this parent cluster)

    \end{itemize}

\end{theorem}

\begin{proof}
    We start by providing an algorithm to build $(2B^j + 1)$-separated weak decomposition from the \cref{theorem:main_alg_gk}. As a reminder, a brief structure of an algorithm is provided below:

    \begin{itemize}
        \item Algorithm consists of constructing $O(\log{n})$ colors, one at a time
        \item Constructing one color involve $O(\log{n})$ phases, corresponding to the number of bits in the identifiers; in each phase we kill some nodes and change labels of some other nodes
        \item One phase consists of $O(\log^2{n})$ steps, each consisting of two substeps:

        \begin{itemize}
            \item Performing a BFS up to depth $(2B^j+1) = O(B^j)$ from all the blue nodes
            \item Collecting the number of proposing nodes in the roots of the Steiner trees, and making corresponding decisions (accept/reject)
        \end{itemize}
    \end{itemize}

    From \cref{theorem:BFSwithcovers} we know how to perform $B^{j-1}$-thresholded BFS. So, we can implement $O(B^j)$-thresholded BFS as $O(\log^3{n})$ rounds of $B^{j-1}$-thresholded BFS, which takes time $O(B^{j-1}\log^{11}{n} \cdot \log^3{n}) = O(B^j\log^{11}{n})$ and has energy cost $O(\log^{21}{n})$. As we perform $O(B^j)$-thresholded BFS $O(\log^4{n})$ times, the total time spent here is $O(B^j\log^{15}{n})$, and the energy cost is $O(\log^{25}{n})$.

    Collecting the number of proposing nodes in Steiner tree roots, making decisions, and propagating them is easier with our already developed convergecast and broadcast approaches. For each such convergecast/broadcast, we need $O(1)$ energy cost per Steiner tree, with $O(B^{j-1}\log^3{n})$ rounds. But, since each edge is in $O(\log^4{n})$ Steiner trees of clusters of $B^j$-cover, we have to divide rounds of this step into mega-rounds of length $O(\log^4{n})$, leading to $O(\log^4{n} \cdot \log^4{n}) = O(\log^{8}{n})$ energy cost per all Steiner trees, and time complexity $O(B^{j-1}\log^3{n} \cdot \log^4{n}) = O(B^{j}\log^4{n})$. As we perform such collecting $O(\log^4{n})$ times, the time spent is $O(B^j\log^{7}{n})$, and the energy cost is $O(\log^{12}{n})$. Thus, overall, the total time complexity is $O(B^j\log^{15}{n})$, and the energy cost is $O(\log^{25}{n})$.

    After constructing the $(2B^j + 1)$-separated weak decomposition, we apply the algorithm from \cref{theorem:synchcover}: we do $O(\log{n})$ BFSs up to depth $B^j$, one from nodes of every color; these BFSs are dominated by the time complexities of constructing this decomposition. 

    Finally, we need to mention how to assign a parent to every cluster of the sparse $B^{j-1}$-cover. For every node $v$, we make it remember the cluster of its color, obtained after expansion by the BFS up to distance $B^j$. Now, for every cluster $C$ of sparse $B^{j-1}$-cover, we assign to it a cluster remembered by the root of its Steiner tree. This cluster contains $B^j$-neighborhood of the root, so it satisfies the required conditions (refer to \cref{observation:clusterparent} for details). Then, we do one round of broadcast from the roots of Steiner trees to the nodes of the trees.
\end{proof}

\subsection{The Complete BFS Algorithm in $\tilde{O}(D)$ Time and $\tilde{O}(1)$ Energy}

\begin{theorem}
    \label[theorem]{theorem:main}
    Let $D$ denote the diameter of the entire graph. There exists an algorithm that correctly computes BFS from scratch, works in time $O(D\log^{18}{n})$, and has energy cost $O(\log^{26}{n})$. 

\end{theorem}

\begin{proof}
    We start by computing sparse $1$-cover and sparse $B$-cover, by simply having all nodes be awake for the entire duration. By \cref{theorem:synchcover}, we can do this with time complexity and energy cost of $O(B\log^{10}{n}) = O(\log^{13}{n})$. Next, we will compute sparse $B^2$-cover, sparse $B^3$-cover, $\ldots$, sparse $B^L$-cover, where $L = \lceil{\log_B D\rceil}$, as in \cref{theorem:coverwithcovers}. Then, with the layered sparse $D$-cover, we will compute the actual BFS, as in \cref{theorem:BFSwithcovers}. 
    
    How do we know when to stop the BFS? Nodes do not know $D$. So, instead, after we constructed sparse $B^j$-cover, we check if any of its clusters contains all the nodes of the graph, and tell all nodes that the phase of constructing sparse covers is over, and they can start the BFS part now. 

    How to check if any cluster of the sparse $B^j$-cover contains all the nodes of the graph? After the construction of $B^j$-cover, Let us have a period of $\log{n}$ rounds, in which all nodes are awake, and all nodes tell all their neighbors the list of clusters of $B^j$-cover, in which they are present. Then, for every cluster $C$ of the sparse $B^j$-cover, every node checks whether all its neighbors told it that they are in $C$. Then, this information is propagated to the root of $C$ with the convergecast. If the root detects that $C$ contains all nodes, it broadcasts this information to them. This would take extra $O(B^j)$ time and extra $O(\log{n}) + O(1) = O(\log{n})$ energy cost, which are dominated by the construction of sparse covers.

    By the time we construct sparse $B^L$-cover, some clusters will definitely contain all the nodes. So, the resulting time complexity is $O(\log^{11}{n}) + O(\log^{15}{n}(B^2 + \ldots + B^L)) = O(\log^{15}{n}B^L)  = O(D\log^{18}{n})$, and the resulting energy cost is $O(\log^{25}{n} \cdot \log{n}) + O(\log^{18}{n}) = O(\log^{26}{n})$, as desired.
\end{proof}

With this idea, we can now compute $\mathcal{D}$-thresholded BFS for any $\mathcal{D}$, with an analogous proof. 

\begin{theorem}
    \label[theorem]{theorem:mainthresholded}
    Let $\mathcal{D}$ be any integer. There exists an algorithm that correctly computes $\mathcal{D}$-thresholded BFS from scratch, works in time $O(\mathcal{D}\log^{18}{n})$, and has energy cost $O(\log^{26}{n})$. 
\end{theorem}

\subsection{Closest-Source Shortest Paths in the \energy Model}

\label{subsec:energymsps}

Now we go back to the original problem, where weights are nonnegative integers in $[0, \poly(n)]$. 

\begin{restatable}{theorem}{MAINCSSP}
\label[theorem]{theorem:main_cssp_energy}
    Consider a graph $G = (V, E)$ and a set of sources $S$. If the weights of all edges in $G$ are nonnegative integers bounded by $O(\poly(n))$, there exists an algorithm that computes all distances $dist(S, v)$, in time $O(n\log^{19}{n})$ and energy cost $O(\log^{26}{n})$.
\end{restatable}


The algorithm is analogous to the algorithm in \cref{subsec:mainAlgo}. The only energy-consuming components are running thresholded BFSs for computing approximations of distances, and computing connected components and their (rooted) spanning trees. We replace the first one with the low-energy thresholded BFS algorithm from \cref{theorem:mainthresholded}, and the second one with the algorithm from \cref{theorem:spanningenergy}. \fullOnly{We defer the complete description of the algorithm and the proof to \cref{app:main_cssp_energy}.}\shortOnly{We defer the complete description of the algorithm and the proof to the full version.}

\section*{Acknowledgement} This work was partially supported by a grant from the Swiss National Science Foundation (project grant 200021$\_$184735). 

\bibliographystyle{alpha}
\bibliography{references, ref2, ref3}

\fullOnly{
\appendix
\clearpage

\section{A missing part of Section \ref{sec:prelim}}
\label{app:prelim}

\begin{proof}[Proof of \Cref{lemma:approx}]

Define $\tau = \frac{\epsilon W}{n}$. For every edge, Let us round up its weight $w(e)$ to the nearest multiple of $\tau:$ $w'(e) = \tau \lceil \frac{w(e)}{\tau} \rceil$. Let us denote this graph with updated weights by $G'$.

Now, let us perform BFS on $G'$, starting with $S$, waiting at the edge with weight $w'(e) = a\tau$ for $a$ rounds. Let us perform it for $k = \lceil \frac{3n}{\epsilon} \rceil$ rounds. Every node at the distance at most $\frac{3n}{\epsilon}\tau = 3W$ from $s$ (in a graph with updated weights) will learn its distance from $S$ in $G'$. 

For some node $t$, consider the shortest path from $S$ to $t$ in $G$: $s = v_0, v_1, \ldots, v_k = t$ for some $s\in S$. In $G'$, we have: 
    
$$dist(S, t) \le dist'(S, t) \le \sum_{i = 0}^{k-1} w'(v_i, v_{i+1}) < \sum_{i = 0}^{k-1} (w(v_i, v_{i+1}) + \tau) < dist(S, t) + n\tau = dist(S, t) + \epsilon W$$

So, for any node $t$ with $dist(S, t) \in [0, 3W- \epsilon W]$, $t$ will learn an approximation $dist'(s, t)$, which satisfies all the constraints from the statement of the lemma. If this BFS doesn't get to a node $v$ in $k$ rounds, it knows that it is at distance at least $2W$ from the sources.
\end{proof}


\section{The missing proof from Theorem \ref{theorem:main_zero}}
\label{app:zero-weight}

\begin{proof}[Proof of \Cref{theorem:main_zero}]
    Let us replace every edge of weight $0$ with an edge of weight $1$, and every edge of weight $w>0$ with an edge of weight $nw$. On this new graph, let us compute the distances from $S$ to all other nodes, with an algorithm from \cref{theorem:main_no_zero} (note that all weights are still $O(\poly(n))$). We will do this in time $O(n\log^2{n})$ and congestion $O(\log^2{n})$. If the $dist(S, v) = x$ in the new graph, the correct distance is simply $\lfloor \frac{x}{n} \rfloor$.
\end{proof}

\section{A review of the network decomposition of Rozhon and Ghaffari}
\label{app:main_alg_gk}

\MAINALGGK*

In the following lemma, we describe the process for constructing the clusters of one color of the network decomposition (e.g., the first color), in a way that it clusters at least half of the vertices.  Since after each application of this lemma only half of the vertices remain, by $\log n$ repetitions, we get a decomposition of all vertices, with $\log n$ colors.

\begin{lemma}\label[lemma]{lemma:main_alg_gk} Let $S\subseteq V$ denote the set of living vertices. There is a deterministic distributed algorithm that, in $O(k \log^9{n})$ communication rounds, finds a subset $S' \subseteq S$ of living vertices, where $|S'|\geq |S|/2$, such that the subgraph $G[S']$ induced by set $S'$ is partitioned into disjoint clusters, so that every two of these clusters are at a distance larger than $k$, and each cluster has weak-diameter $O(k \cdot \log^3 n)$ in graph $G$.

Moreover, for each such cluster $\mathcal{C}$, we have a Steiner tree $T_\mathcal{C}$ with radius $O(k\cdot \log^3 n)$ in $G$ where all nodes of $\mathcal{C}$ are exactly the terminal nodes of $T_\mathcal{C}$. Furthermore, each edge in $G$ is in $O(\log^3{n})$ of these Steiner trees.
\end{lemma}

We obtain \Cref{theorem:main_alg_gk} by $c=\log n$ iterations of applying \Cref{lemma:main_alg_gk}, starting from $S=V$. For each iteration $j\in [1, \log n]$, the set $S'$ are exactly nodes of color $j$ in the network decomposition, and we then continue to the next iteration by setting $S\gets S\setminus S'$. 

In the proof of the lemma, the following observation is useful. Once again, its proof is omitted.

\paragraph{Algorithm for \cref{lemma:main_alg_gk}} The construction has $b=O(\log n)$ phases, corresponding to the number of bits in the identifiers. Initially, we think of all nodes of $S$ as \textbf{living}. During this construction, some living nodes \textbf{die}. We use $S'_i$ to denote the set of living vertices at the beginning of phase $i\in [0, b-1]$. Slightly abusing the notation, we let $S'_b$ denote the set of living vertices at the end of phase $b-1$ and define $S'$ to be the final set of living nodes, i.e., $S' := S'_{b}$. 

Moreover, we label each living node $v$ with a $b$-bit string $\ell(v)$, and we use these labels to define the clusters. At the beginning of the first phase, $\ell(v)$ is simply the unique identifier of node $v$. This label can change over time. For each $b$-bit label $L\in \{0,1\}^b$, we define the corresponding cluster $S'_i(L)\subseteq S'_i$ in phase $i$ to be the set of all living vertices $v\in S'_i$ such that $\ell(v)=L$. 
We will maintain one Steiner tree $T_L$ for each cluster $S'_i(L)$ where all nodes $S'_i(L)$ are the terminal nodes of $T_L$. 
Initially, each cluster consists of only one vertex and this is also the only (terminal) node of its respective Steiner tree. 

\medskip
\paragraph{Construction invariants.} The construction is such that, for each phase $i\in [0, b-1]$, we maintain the following invariants: 
\begin{enumerate}
    \item[(I')] For each $i$-bit string $Y$, the set $S'_i(Y)\subseteq S'_i$ of all living nodes whose label ends in suffix $Y$ has no other living nodes $S'_i\setminus S'_i(Y)$ in its $k$-hop neighbourhood. 
    
    In other words, the set $S'_i(Y)$ is a union of some connected components of the subgraph $G[S'_i]$ induced by living nodes $S'_i$ and in the $k$-hop neighbourhood in $G$ around $S'_i(Y)$ all nodes are either dead or they do not belong to the set $S$ (they were colored by previous application of the algorithm). 
    \item[(II')] For each label $L$ and the corresponding cluster $S'_i(L)$, the related Steiner tree $T_L$ has radius at most $i \cdot k \cdot R$, where $R=O(\log^2 n)$. 
    \item[(III)] We have $|S'_{i+1}|\geq |S'_i|(1-1/2b)$.
\end{enumerate}

These invariants, together with the observation that each edge is used in $O(\log n)$ Steiner trees, prove \Cref{lemma:main_alg_gk}.

In particular, from the first invariant, we conclude that at the end of $b$ phases, different clusters are at a distance at least $k+1$ from each other. 
From the second invariant, we conclude that each cluster has a Steiner tree with radius $bR=O(k\log^3{n})$. 
Finally, from the third invariant, we conclude that for the final set of living nodes $S'= S'_{b}$, we have $|S'| \geq (1-1/2b)^{b} |S| \geq |S|/2$.

\medskip
\paragraph{Outline of one phase of construction.} We now outline the construction of one phase and describe its goal. Let us think about some fixed phase $i$. We focus on one specific $i$-bit suffix $Y$ and the respective set $S'_i(Y)$. 
Let us categorize the nodes in $S'_i(Y)$ into two groups of \textbf{blue} and \textbf{red}, based on whether the $(i+1)^{th}$ least significant bit of their label is $0$ or $1$. Hence, all blue nodes have labels of the form $(*\ldots*0Y)$, and all red nodes have labels of the form $(*\ldots*1Y)$, where $*$ can be an arbitrary bit. During this phase, we make some small number of the red vertices die, and we change the labels of some of the other red vertices to blue labels (and then the node is also colored blue). All blue nodes remain living and keep their label. The eventual goal is that, at the end of the phase, among the living nodes, there is no blue node $b$ and red node $r$ with $dist(b, r) \le k$. This leads to invariant (I) for the next phase. The construction ensures that we kill at most $|{S'}_i(Y)|/{2b}$ red vertices of set $S'_i(Y)$, during this phase. We next describe this construction.

\medskip
\paragraph{Steps of one phase.}
Each phase consists of $R = 10b\log n = O(\log^2 n)$ steps, each of which will be implemented in $O(k\cdot \log^6 n))$ rounds. Hence, the overall round complexity of one phase
is $O(k\log^8{n})$ and over all the $O(\log n)$ phases, the round complexity of the whole construction of \Cref{lemma:main_alg_gk} is $O(k\cdot \log^9 n)$ as advertised in its statement. 
Each step of the phase works as follows: all blue nodes start a BFS from them up to distance $k$; a node that is reached in the BFS is added to the cluster of the node from which it received the first ``join'' proposal. Next, all red nodes that were added to the cluster will try to join it, to adopt its label. 

For each blue cluster $A$, we have two possibilities: 
\begin{enumerate}
\item[(1)] If the number of adjacent red nodes that requested to join $A$ is less than or equal to $|A|/2b$, then $A$ does not accept any of them and all these requesting red nodes die (because of their request being denied by $A$). In that case, cluster $A$ \textbf{stops} for this whole phase and does not participate in any of the remaining steps of this phase. 
\item[(2)] Otherwise --- i.e., if the number of adjacent red nodes that requested to join $A$ is strictly greater than $|A|/2b$ --- then $A$ accepts all these requests and each of these red nodes change their label to the blue label that is common among all nodes of $A$. In this case, we also grow the Steiner tree of cluster $A$ by one hop to include all these newly joined nodes.
\end{enumerate}

After the breadth first search algorithm finishes, roots of all Steiner trees collect the number of proposing red nodes and each root decides to either accept all proposing red vertices and recolor them to blue, or it makes them die and stops growing. 
The Steiner trees, however, stay the same even if some of its vertices die, the red nodes that died are just labeled as nonterminals. This finishes the description of one step of current phase. 

As each edge is in $O(\log^3{n})$ cluster trees, and diameters of the clusters are $O(k\log^3{n})$, one step can be implemented in $O(k \log^6{n})$ iterations. As there are $O(\log^2{n})$ steps per phase and $O(\log{n})$ phases per color, we get the resulting runtime of $O(k \log^9{n})$.

\section{The complete closest-source shortest paths algorithm in \energy model}
\label{app:main_cssp_energy}

The main goal of this section is to provide a detailed proof of \cref{theorem:main_cssp_energy}:

\MAINCSSP*

Our algorithm will closely resemble the algorithm from \cref{subsec:mainAlgo}. The goal of this section is to show how to perform each step of that algorithm in the \energy model.

We also assume that the weights are positive and in $[1, 2, \ldots, \poly{n}]$. We handle the case of edges with weight $0$ in the same way as in \cref{theorem:main_zero}.  

\subsection{CSSP approximation in \energy model}

Let us denote the source by $s$. 

\begin{lemma}
    \label[lemma]{lemma:approxenergy}
    Consider a graph $G = (V, E)$ and a set of sources $S$. There is an algorithm in \energy model that, given $\epsilon \in (0, 1)$ and an integer $W>0$, for each node $v$ outputs $dist'(S, v)$, such that:

    \begin{itemize}
        \item If $dist'(S, v) \neq \infty$, then $dist(S, v) \le dist'(S, v) < dist(S, v) + \epsilon W$
        
        \item If $dist'(S, v) = \infty$, then $dist(S, v) > 2W$

        \item It works in time $O(\frac{n\log^{18}{n}}{\epsilon})$ and has energy cost $O(\log^{25}{n})$.
    \end{itemize}

\end{lemma}

\begin{proof}

As in the proof of \cref{lemma:approx}, we define $\tau = \frac{\epsilon W}{n}$. For every edge, let us round up its weight $w(e)$ to the nearest multiple of $\tau:$ $w'(e) = \tau \lceil \frac{w(e)}{\tau} \rceil$. Let us denote this graph with updated weights by $G'$. If we denote distances in graph $G'$ by $dist'$, for every node $t$ we have 

$$dist(S, t) \le dist'(S, t) < dist(S, t) + \epsilon W$$

So, the only question is how to perform $k$-thresholded BFS on $G'$, where $k = \lceil \frac{3n}{\epsilon} \rceil = O(\frac{n}{\epsilon})$. For an edge $e$ with new weight $w'(e) = a\tau$ for $a>0$, Let us create $a-1$ imaginary nodes on it, which subdivide $e$ into $a$ edges of weight $1$. In this graph all edges have weight $1$, so we just have to perform a regular $O(\frac{n}{\epsilon})$ thresholded CSSP. Finally, note that endpoints of edge $e$ can simulate the behavior of the imaginary nodes (as these imaginary only ever send messages to each other end to the endpoints), so they don't affect runtime or energy cost. Hence, according to \cref{theorem:mainthresholded}, we get an algorithm for this which works in time $O(\frac{n\log^{18}{n}}{\epsilon})$ and has energy cost $O(\log^{25}{n})$.

\end{proof}


\subsection{The actual algorithm for CSSP in \energy model}

We will describe an algorithm for $\mathcal{D}$-thresholded CSSP in the \energy model. 
To compute CSSP we would then simply run this $\mathcal{D}$-thresholded CSSP for $\mathcal{D} = 2^L \ge n\cdot maxW$.

\begin{enumerate}

    \item If $\mathcal{D} = 1$, we are in the base case. In this case, the only nodes with $dist(S, v) \le \mathcal{D}$ are the sources themselves and nodes that are connected to some source by an edge of weight $1$. All nodes can detect this in one round. If $\mathcal{D}>1$, proceed to the next steps.

    \item Compute all connected components of $G$, and a spanning tree for each of them, with an algorithm from \cref{theorem:spanningenergy}. We will be solving this problem for each component independently.
    
    \item Choose $\epsilon = 0.5$, and compute an approximation $dist'(S, v)$ of distances from all nodes to $S$ with an algorithm from \cref{lemma:approxenergy}. Let $V_1$ denote the set of nodes $v$ with $dist'(S, v) <\mathcal{D} + \epsilon \mathcal{D}$. Clearly, for any $v \in V_1$ we have $dist(S, v) < \mathcal{D} + \epsilon \mathcal{D}$. Also, for any $v \in V$ with $dist(S, v) \le \mathcal{D}$ we have $dist'(S, v) < \mathcal{D} + \epsilon \mathcal{D}$, so $v \in V_1$.
    
    \item Let $\mathcal{D}_1 = \frac{\mathcal{D}}{2}$. Perform a $\mathcal{D}_1$-thresholded CSSP for nodes $V_1$, with set $S$ of sources. The question becomes, how do the nodes of a particular connected component $C$ learn that all nodes of $C$ which are in $V_1$ are done with this $\mathcal{D}_1$-thresholded CSSP, as we cannot do simple convergecast. 
    
    The solution is simple: in step $2$, we will also tell all nodes the sizes of their connected components. Then, we will collect information about the completion of $\mathcal{D}_1$-thresholded CSSP via the spanning tree of each component $C$, with period $|C|$, as described in \cref{subsubsec:collecting}. This way, all nodes agree on a time when they will proceed to the next step.
    
    \item 
    Let $V_2$ denote the set of all nodes $v$ with $d(S, v) \le \mathcal{D}_1$ (after previous steps, all nodes know whether they are in $V_2$). For every edge $(v, u)$ with $v \in V_2, u \in V_1 \setminus V_2$, create an imaginary node $x_{vu}$ somewhere on the edge $(u, v)$, splitting it into two edges $(v, x_{vu})$ and $(x_{vu}, u)$, so that $w((v, x_{vu})) = \mathcal{D}_1 - dist(S, x)$. Let us denote this set of imaginary nodes as $X$. These nodes from $X$ form a ``cut'' at distance $\mathcal{D}_1$ from $S$. 

    \item Finally, perform $\mathcal{D}_1$-thresholded CSSP for nodes from $X\cup (V_1 \setminus V_2)$ with set $X$ of sources. For every imaginary node $x_{vu}$, node $u \in V_1 \setminus V_2$ can simulate the messages of $x_{vu}$ (it only ever sends messages to $u$). 

    For any node $v \in V_1 \setminus V_2$, $dist(X, v) = dist(S, v) + D_1$. As $\mathcal{D} = 2\mathcal{D}_1$, then, if $dist(X, v) \le \mathcal{D}_1$, we have $dist(S, v) = \mathcal{D}_1 + dist(X, v) \le \mathcal{D}$, and if $dist(X, v) > \mathcal{D}_1$, we have $dist(S, v) = \mathcal{D}_1 + dist(X, v) > \mathcal{D}$, hence the algorithm computes correct output.
\end{enumerate}

\subsection{Analysis of the algorithm}

As in \cref{subsec:analysis}, Let us denote the running time of $\mathcal{D}$-thresholded BFS for a set of nodes $V$ and a set of sources $S$ by $T(V, S, \mathcal{D})$. Then, the recurrence from \cref{subsec:mainAlgo}, keeping all the notations becomes:

$$T(V, S, \mathcal{D}) = \underbrace{O(|V|\log^2{|V|})}_\text{Connected components} + \underbrace{O(|V|\log^{18}{|V|})}_\text{Approximation} +  \underbrace{T(V_1, S, \mathcal{D}_1)}_\text{Call on $V_1$} +  \underbrace{O(|V|)}_\text{Convergecast} +  \underbrace{T(X \cup (V_1 \setminus V_2), X, \mathcal{D}_1)}_\text{Call on $V_1 \setminus V_2$}$$

This simplifies to

$$T(V, S, \mathcal{D}) = O(|V|\log^{18}{V}) + T(V_1, S, \mathcal{D}_1) + T(V_1 \setminus V_2, X, \mathcal{D}_1)$$

\cref{lemma:appear} and \cref{corollary:appearSum} still hold, so we can proceed to the main result.


\begin{proof}[Proof of \cref{theorem:main_cssp_energy}]

Let us run $T(V, S, \mathcal{D})$ with $\mathcal{D} = 2^L$, where $L$ is the smallest integer such that $2^L \ge n\cdot maxW$ is an upper bound on the distances from $S$ to other nodes. Since all weights are $O(\poly(n))$, $O(\log{\mathcal{D}}) = O(\log{n})$. 

Now we can analyze the recursion: since, by \cref{corollary:appearSum}, the sum of sizes of $|V'|$ over all subproblems is $O(|V|\log{\mathcal{D}})$, the sum of $O(|V'|\log^{18}{V'})$ is $O(|V|\log{\mathcal{D}}\log^{18}{|V|})$. Therefore, $T(V, S, \mathcal{D}) = O(|V|\log{\mathcal{D}}\log^{18}{V}) = O(|V|\log^{19}{V})$.

Now let us analyze energy cost. Every node is involved in only $O(\log{\mathcal{D}})$ subproblems, in each it has to spend up to $O(\log^{26}{n})$ energy. Additionally, in step $4$ of the algorithm, a node may spend up to $O(\log^{19}{n})$ energy waiting until its component is done. This, however, is dominated by the $O(\log^{26}{n})$ term.

Hence, the total energy cost per node is $O(\log^{26}{n})$.

\end{proof}
}

\end{document}